\documentclass[submission,copyright,creativecommons]{eptcs}
\providecommand{\event}{Non-Classical Logics 2024 (NCL'24)} 

\usepackage{iftex}

\ifpdf
  \usepackage{underscore}         
  \usepackage[T1]{fontenc}        
\else
  \usepackage{breakurl}           
\fi

\usepackage[T1]{fontenc}
\usepackage{bussproofs}
\usepackage{amsthm} 
\usepackage{MnSymbol}
\EnableBpAbbreviations

\title{Incomplete Descriptions and Qualified Definiteness}
\author{Bartosz Wi\k{e}ckowski
\institute{Goethe University\thanks{Support by the Deutsche Forschungsgemeinschaft (DFG), grant WI 3456/5-1, is gratefully acknowledged.}
\\ Frankfurt am Main, Germany}
\email{wieckowski@em.uni-frankfurt.de}
}
\def\titlerunning{Incomplete descriptions and qualified definiteness}
\def\authorrunning{B. Wi\k{e}ckowski}

\begin{document}
\theoremstyle{plain} 
\theoremstyle{definition} 
\theoremstyle{remark} 
\newtheorem{convention}{Convention}[section]
\newtheorem{corollary}{Corollary}[section]
\newtheorem{definition}{Definition}[section]
\newtheorem{digression}{Digression}[section]
\newtheorem{example}{Example}[section]
\newtheorem{illustration}{Illustration}[section]
\newtheorem{lemma}{Lemma}[section]
\newtheorem{proposition}{Proposition}[section]
\newtheorem{theorem}{Theorem}[section]
\newtheorem{remark}{Remark}[section]

\maketitle

\begin{abstract}
According to Russell, strict uses of the definite article \textquoteleft the\textquoteright\space 
in a definite description \textquoteleft the $F$\textquoteright\space involve uniqueness; 
in case there is more than one $F$, \textquoteleft the $F$\textquoteright\space 
is used somewhat loosely, and an indefinite description \textquoteleft an $F$\textquoteright\space 
should be preferred. 
We give an account of constructions of the form \textquoteleft the $F$ is $G$\textquoteright\space 
in which the definite article is used loosely 
(and in which \textquoteleft the $F$\textquoteright\space is, therefore, incomplete), 
essentially by replacing the usual notion of identity in Russell's uniqueness clause 
with the notion of qualified identity,    
i.e., \textquoteleft $a$ is the same as $b$ in all $\mathcal{Q}$-respects\textquoteright, 
where $\mathcal{Q}$ is a subset of the set of predicates $\mathcal{P}$. 
This modification gives us qualified notions of uniqueness and definiteness. 
A qualified definiteness statement   
\textquoteleft the $\mathcal{Q}$-unique $F$ is $G$\textquoteright\space  
is strict in case $\mathcal{Q} = \mathcal{P}$ and 
loose in case $\mathcal{Q}$ is a proper subset of $\mathcal{P}$. 
The account is made formally precise in terms of proof theory and proof-theoretic semantics. 

\end{abstract}

\section{Introduction}

Sometimes we use the definite description \textquoteleft the $F$\textquoteright\space 
in cases in which there is a unique $F$.  
According to Russell (\cite{Russell}: 481), the definite article 
\textquoteleft the\textquoteright\space is used strictly in such cases.  
For example, speaking about Francis, 
we use \textquoteleft the pope\textquoteright\space in (1.1) in this way. 
\begin{itemize}
\item[(1.1)] The pope is bald. 
\end{itemize}
Sometimes, as Russell notes, we use \textquoteleft the $F$\textquoteright\space also in cases, 
in which there is more than one $F$. 
For example,  
\textquoteleft the bishop\textquoteright\space in (1.2) 
is used in this loose way (as would be \textquoteleft the pope\textquoteright\space during a schism). 
\begin{itemize}
\item[(1.2)] The pope blesses the bishop.  
\end{itemize}
According to Russell, such loose uses of \textquoteleft the $F$\textquoteright\space 
should be avoided in favour of the indefinite description \textquoteleft an $F$\textquoteright. 

In this paper, we propose a formal account of both uses of 
\textquoteleft the $F$\textquoteright\space in terms of qualified definiteness. 
On a Russellian analysis, a construction of the form 
\textquoteleft the $F$ is $G$\textquoteright\space  
is explained in terms of an existence, a uniqueness, and a predication clause: 
\begin{itemize}
\item[(E)] There is at least one $F$. 
\item[(U)] There is at most one $F$. 
\item[(P)] Every $F$ is $G$. 
\end{itemize}
We modify this analysis mainly by replacing the usual notion of identity in the 
definition of uniqueness with the notion of \textit{qualified identity} proposed in \cite{negpredd},   
i.e., \textquoteleft $a$ is the same as $b$ in all $\mathcal{Q}$-respects\textquoteright, 
where $\mathcal{Q}$ is a subset of the set of predicates $\mathcal{P}$. 
The notion of \textit{qualified uniqueness} that results from this replacement says: 
\begin{itemize}
\item[(QU)] For every $x$ and $y$, if they are $F$, then they are identical with respect to every predicate in $\mathcal{Q}$. 
\end{itemize}
Finally, a statement of \textit{qualified definiteness} says, 
combining the three Russellian components:  
\begin{itemize}
\item[(QD)]  The $\mathcal{Q}$-unique $x$ which is $F$ is $G$. 
\end{itemize}
Qualified definiteness, unlike standard definiteness, allows for fine-tuning. 
Let $\mathcal{Q}^{\prime}$ be a proper subset of $\mathcal{P}$ 
(i.e., $\mathcal{Q}^{\prime} \subset \mathcal{P}$). 
If $\mathcal{Q} = \mathcal{P}$ in (QD), then we get the reading 
\textquoteleft the \textit{only} $x$ which is $F$ is $G$\textquoteright.  
We may use this reading only in case there is a single $x$ that is $F$. 
This is definiteness proper. 
If, on the other hand, we put $\mathcal{Q} = \mathcal{Q}^{\prime}$, then we get: 
\textquoteleft the $x$ which is $F$ is $G$\textquoteright.  
We may use this reading only in case there are at least two things 
which are $F$ that are indiscernible with respect to $\mathcal{Q}^{\prime}$, 
but discernible with respect to $\mathcal{P}\setminus\mathcal{Q}^{\prime}$. 
This is restricted definiteness. 
What is subject to restriction, on this account, is thus the set 
of $\mathcal{Q}$-respects (rather than, e.g., a domain of quantifiers \cite{QDR}).   

Below, we provide the details of this proposal. 
It will differ from competing semantic analyses of incomplete descriptions also 
in that it will be couched in a framework of proof-theoretic semantics 
(see \cite{PSH} for an overview) 
rather than in some version of model-theoretic semantics. 
(For an overview of the literature on incomplete descriptions 
see, e.g., \cite{Ludlow}: sect. 5.3. 
An elaborate model-theoretic account is \cite{Elbourne}.) 

Sect. 2 defines the formal language.  
Sect. 3 recapitulates the relevant fragment of the intuitionistic bipredicational 
natural deduction systems 
defined in \cite{negpredd} and combines it with the rules for definiteness proposed 
in \cite{iota1}, \cite{iota2} into proof systems for qualified definiteness,  
establishing normalization and the subexpression (and subformula) property for them. 
Sect. 4 defines a proof-theoretic semantics for qualified definiteness, and 
Sect. 5 applies this semantics to incomplete descriptions in the manner suggested above. 
The paper ends with a brief outlook in Sect. 6.

\section{The language}

We extend the bipredicational 
language $\mathcal{L}$ motivated and defined in \cite{negpredd} 
with contextually defined operators for qualified definiteness  
and call the extended language $\mathcal{L}\iota$.

$\mathcal{L}$ is a first-order language. 
It is bipredicational, since it allows for both predication and predication failure. 
We first recapitulate those parts of its definition which are relevant for present purposes. 

\begin{definition}
$\mathcal{C}$ is the set of individual (or nominal) constants (form: $\alpha_{i}$) 
and 
$\mathcal{P}$ is the set of $n$-ary predicate constants (form: $\varphi^{n}_{i}$) 
of $\mathcal{L}$.  
Moreover, $Atm$ is the set of atomic sentences (form: $\varphi^{n}\alpha_{1} ... \alpha_{n}$) 
of $\mathcal{L}$. 
$Atm(\alpha) =_{def} \{A \in Atm: A$ contains at least one occurrence of $\alpha \in \mathcal{C}\}$  
and 
$Atm(\varphi^{n}) =_{def} \{A \in Atm: A$ contains an occurrence of $\varphi^{n} \in \mathcal{P}\}$. 
	A nominal term $o_{i}$ is either a nominal constant or a nominal variable $x_{i}$. 
Atomic formulae have the form $\varphi^{n}o_{1} ... o_{n}$ and are used for predication. 
Negative predications (or predication failures) 
take the form $-\varphi^{n}o_{1} ... o_{n}$ 
(reading: 
\textquoteleft the ascriptive combination of 
$\varphi^{n}$ with $o_{1}, ..., o_{n}$ fails\textquoteright). 
\end{definition}

\begin{definition}
Defined symbols of $\mathcal{L}$:  
\begin{enumerate}
\item $\neg A =_{def} A\supset\bot$ (negation) 
\item $A \leftrightarrow B =_{def} (A\supset B) \& (B\supset A)$ (equivalence)  
\item Let $\varphi^{n}$ be an $n$-ary predicate constant. 
\begin{quote}
$P^{n}_{\varphi^{n}}(o_{1}, o_{2}) =_{def}$  
\begin{quote}
$\forall z_{1} ... \forall z_{n-1}\forall z_{n}$ 
$((\varphi^{n}o_{1}z_{2} ... z_{n}$ 
$\leftrightarrow$ 
$\varphi^{n}o_{2}z_{2} ... z_{n})$\\
$\&$
$(\varphi^{n}z_{1}o_{1} ... z_{n}$ 
$\leftrightarrow$ 
$\varphi^{n}z_{1}o_{2} ... z_{n})$\\
$\&$ ... $\&$ 
$(\varphi^{n}z_{1} ... z_{n-1}o_{1}$ 
$\leftrightarrow$ 
$\varphi^{n}z_{1} ... z_{n-1}o_{2}))$ 
\end{quote}
\end{quote}
\begin{quote}
$N^{n}_{\varphi^{n}}(o_{1}, o_{2}) =_{def}$  
\begin{quote}
$\forall z_{1} ... \forall z_{n-1}\forall z_{n}$ 
$((-\varphi^{n}o_{1}z_{2} ... z_{n}$ 
$\leftrightarrow$ 
$-\varphi^{n}o_{2}z_{2} ... z_{n})$\\
$\&$
$(-\varphi^{n}z_{1}o_{1} ... z_{n}$ 
$\leftrightarrow$ 
$-\varphi^{n}z_{1}o_{2} ... z_{n})$\\
$\&$ ... $\&$ 
$(-\varphi^{n}z_{1} ... z_{n-1}o_{1}$ 
$\leftrightarrow$ 
$-\varphi^{n}z_{1} ... z_{n-1}o_{2}))$ 
\end{quote}
\end{quote}
Let $\varphi^{k_{1}}_{1}, ... , \varphi^{k_{m}}_{m}$ be all the predicate constants 
in $\mathcal{Q}$, where $\varphi_{i}$ is $k_{i}$-ary and $\mathcal{Q}\subseteq\mathcal{P}$. 
\begin{quote}
\textit{Positive qualified identity}:
\begin{quote}
$o_{1}\overset{+}{=}_{\mathcal{Q}}o_{2} =_{def} 
P^{k_{1}}_{\varphi_{1}}
(o_{1}, o_{2})$ 
$\&$ ... $\&$
$P^{k_{m}}_{\varphi_{m}}
(o_{1}, o_{2})$\\
(\textquoteleft $o_{1}$ is the same as $o_{2}$ in all $\mathcal{Q}$-respects\textquoteright)
\end{quote}

\textit{Negative qualified identity}:
\begin{quote}
$o_{1}\overset{-}{=}_{\mathcal{Q}}o_{2} =_{def} 
N^{k_{1}}_{\varphi_{1}}
(o_{1}, o_{2})$ 
$\&$ ... $\&$
$N^{k_{m}}_{\varphi_{m}}
(o_{1}, o_{2})$\\
(\textquoteleft $o_{1}$ is the same as $o_{2}$ in no $\mathcal{Q}$-respect\textquoteright)
\end{quote}
\end{quote}
\end{enumerate}
\end{definition}

\begin{remark}
Note that, in contrast to $\neg$, the operator for predication failure $-$ is primitive. 
Moreover, unlike the former, it is sensitive to the internal structure of the formula to 
which it is prefixed. 
\end{remark}

$\mathcal{L}\iota$ extends $\mathcal{L}$ with operators for qualified definiteness by 
adapting the definitions from \cite{iota1}, \cite{iota2}. 
\begin{definition}
We write $\varphi(x)$, suppressing the arity of $\varphi$,   
for atomic formulae $\varphi^{n}o_{1} ... o_{n}$ 
containing (possibly multiple occurrences of) $x$. 
Let $\mathcal{Q} \subseteq \mathcal{P}$. 

\begin{enumerate}
\item \textit{Positive qualified definiteness}: 

	$\psi(\iota_{\mathcal{Q}} x\varphi(x)) =_{def}$ 
	$\exists x \varphi(x)$  
	$\&$  
	$\underbrace{\forall u \forall v((\varphi(u)$ $\&$ $\varphi(v)) \supset u \overset{+}{=}_{\mathcal{Q}} 
	v)}_{\textnormal{\textit{Positive qualified uniqueness}}}$  
	$\&$ 
	$\forall w(\varphi(w) \supset \psi(w))$
	
	(\textquoteleft the $\mathcal{Q}$-unique $x$ which is $\varphi$ is $\psi$\textquoteright; 
	simpler: 
	\textquoteleft the $\mathcal{Q}$-unique $\varphi$ is $\psi$\textquoteright)

\item \textit{Negative qualified definiteness}: 

	$\psi(\iota_{\mathcal{Q}} x-\varphi(x)) =_{def}$ 
	$\exists x -\varphi(x)$  
	$\&$  
	$\underbrace{\forall u \forall v((-\varphi(u)$ $\&$ $-\varphi(v)) \supset u \overset{-}{=}_{\mathcal{Q}} 
	v)}_{\textnormal{\textit{Negative qualified uniqueness}}}$  
	$\&$ 
	$\forall w(-\varphi(w) \supset \psi(w))$
	
	(\textquoteleft the $\mathcal{Q}$-unique $x$ which fails to be $\varphi$ is $\psi$\textquoteright; 
	simpler: 
	\textquoteleft the $\mathcal{Q}$-unique $-\varphi$ is $\psi$\textquoteright)
	\end{enumerate}
\end{definition}

\begin{remark}
The definition of positive qualified definiteness differs from the definition of definiteness 
proposed in \cite{iota1}, \cite{iota2}, in that it does not make use of the familiar primitive notion 
of identity in the uniqueness part. 
In this respect, it significantly departs also from the tradition. 
\end{remark}

Qualified definiteness allows for degrees. 
\begin{definition}
Let $\mathcal{Q}^{\prime} \subset \mathcal{P}$.  
It has  
(i) the highest degree of definiteness 
in case $\mathcal{Q} = \mathcal{P}$ and  
(ii) a lower degree, 
in case $\mathcal{Q} = \mathcal{Q}^{\prime}$.  
Given $\mathcal{Q}^{\prime} \subset \mathcal{P}$, 
we can make the following distinction: 
\begin{enumerate}
\item \textit{Maximal definiteness}:
	\begin{enumerate}
	\item $\psi(\iota_{\mathcal{P}} x\varphi(x))$: 
		\textquoteleft the only $x$ which is $\varphi$ is $\psi$\textquoteright; 
	\item $\psi(\iota_{\mathcal{P}} x-\varphi(x))$: 
		\textquoteleft the only $x$ which fails to be $\varphi$ is $\psi$\textquoteright. 
	\end{enumerate}

\item \textit{Restricted definiteness}: 
	\begin{enumerate}
	\item $\psi(\iota_{\mathcal{Q}^{\prime}} x\varphi(x))$: 
		\textquoteleft the $x$ which is $\varphi$ is $\psi$\textquoteright;  
	\item $\psi(\iota_{\mathcal{Q}^{\prime}} x-\varphi(x))$: 
		\textquoteleft the $x$ which fails to be $\varphi$ is $\psi$\textquoteright. 
	\end{enumerate}
\end{enumerate}
\end{definition}
A loosely used definite description \textquoteleft the $F$\textquoteright\space 
is, thus, construed as a restriction of 
a strictly used \textquoteleft the $F$\textquoteright\space 
(i.e., the maximally definite description \textquoteleft the only $F$\textquoteright). 

\begin{definition}
Negative predications with qualified definite descriptions take 
the following forms:
\begin{enumerate}
\item $-\psi(\iota_{\mathcal{Q}} x\varphi(x))$: 
\textquoteleft the $\mathcal{Q}$-unique $x$ which is $\varphi$ fails to be $\psi$\textquoteright; 

\item $-\psi(\iota_{\mathcal{Q}} x-\varphi(x))$: 
\textquoteleft the $\mathcal{Q}$-unique $x$ which fails to be $\varphi$ fails to be $\psi$\textquoteright. 
\end{enumerate} 
\end{definition}

\section{Proof systems}

In order to obtain a proof system for reasoning with qualified definiteness, 
we enrich the intuitionistic bipredicational $\textbf{I0}(\mathcal{S}^{=}_{b})$-systems 
defined in \cite{negpredd} with rules for qualified definiteness,   
by adapting the rules for definiteness presented in \cite{iota1}, \cite{iota2}.   
We call the resulting systems $\textbf{I0}(\mathcal{S}^{=}_{b})\iota$-systems.

\subsection{Bipredicational natural deduction}

We first repeat the parts of the definition of $\textbf{I0}(\mathcal{S}^{=}_{b})$-systems 
from \cite{negpredd} which are relevant for present purposes. 

\subsubsection{Bipredicational subatomic systems}

\begin{definition}
A \textit{bipredicational subatomic system} 
$\mathcal{S}_{b}$ is a pair 
$\langle\mathcal{I}, \mathcal{R}_{b}\rangle$, where 
$\mathcal{I}$ is a \textit{subatomic base} and  
$\mathcal{R}_{b}$ is a set of \textit{introduction and elimination rules 
for atomic sentences and negative predications}. 
$\mathcal{I}$ is a 3-tuple 
$\langle\mathcal{C}, \mathcal{P}, v\rangle$, where 
$v$ is such that: 
\begin{enumerate}
\item For any $\alpha\in\mathcal{C}$,  
$v: \mathcal{C}\rightarrow \wp(Atm)$, 
where $v(\alpha)\subseteq Atm(\alpha)$. 

\item For any $\varphi^{n}\in\mathcal{P}$, 
$v: \mathcal{P}\rightarrow \wp(Atm)$,  
where 
$v(\varphi^{n})\subseteq Atm(\varphi^{n})$. 
\end{enumerate}
We let $\tau\Gamma =_{def} v(\tau)$ for any $\tau \in \mathcal{C}\cup\mathcal{P}$, 
and call $\tau\Gamma$ the set of \textit{term assumptions} for $\tau$. 
	$\mathcal{R}_{b}$ contains I/E-rules of the following form:
\begin{center}
{\small{
\AXC{$\mathcal{D}_{0}$}
\noLine
\UIC{$\varphi^{n}_{0}\Gamma$}
\AXC{$\mathcal{D}_{1}$}
\noLine
\UIC{$\qquad\alpha_{1}\Gamma$  $\quad ...$} 
\AXC{$\mathcal{D}_{n}$}
\noLine
\UIC{$\alpha_{n}\Gamma$}
\RightLabel{($as$I)}
\TIC{$\varphi^{n}_{0}\alpha_{1} ... \alpha_{n}$}
\DP\qquad
\AXC{$\mathcal{D}_{1}$}
\noLine
\UIC{$\varphi^{n}_{0}\alpha_{1} ... \alpha_{n}$}
\RightLabel{($as$E$_{i}$)}
\UIC{$\tau_{i}\Gamma$}
\DP

\bigskip

\AXC{$\mathcal{D}_{0}$}
\noLine
\UIC{$\varphi^{n}_{0}\Gamma$}
\AXC{$\mathcal{D}_{1}$}
\noLine
\UIC{$\qquad\alpha_{1}\Gamma$  $\quad ...$} 
\AXC{$\mathcal{D}_{n}$}
\noLine
\UIC{$\alpha_{n}\Gamma$}
\RightLabel{($-as$I)}
\TIC{$-\varphi^{n}_{0}\alpha_{1} ... \alpha_{n}$}
\DP\qquad
\AXC{$\mathcal{D}_{1}$}
\noLine
\UIC{$-\varphi^{n}_{0}\alpha_{1} ... \alpha_{n}$}
\RightLabel{($-as$E$_{i}$)}
\UIC{$\tau_{i}\Gamma$}
\DP
}}
\end{center}
Side conditions:
\begin{enumerate}
\item $as$I: 
	$\varphi^{n}_{0}\alpha_{1} ... \alpha_{n} \in 
	\varphi^{n}_{0}\Gamma \cap \alpha_{1}\Gamma \cap ... \cap \alpha_{n}\Gamma$.  

\item $-as$I:
	$\varphi^{n}_{0}\alpha_{1} ... \alpha_{n} \not\in 
	\varphi^{n}_{0}\Gamma \cap \alpha_{1}\Gamma \cap ... \cap \alpha_{n}\Gamma$.  

\item $as$E$_{i}$ and $-as$E$_{i}$:
	$i \in \{0, ..., n\}$ and $\tau_{i} \in \{\varphi^{n}_{0}, \alpha_{1}, ..., \alpha_{n}\}$.
\end{enumerate}
Terminology: We say that 
$-\varphi^{n}_{0}\alpha_{1} ... \alpha_{n}$ is 
\textit{negatively contained} in 
$\varphi^{n}_{0}\Gamma \cap \alpha_{1}\Gamma \cap ... \cap \alpha_{n}\Gamma$, 
in case the side condition on $-as$I is satisfied. 
\end{definition}

\begin{definition}
\textit{Derivations in $\mathcal{S}_{b}$-systems}. 

\textit{Basic step}. 
Any term assumption $\tau\Gamma$, any atomic sentence (resp. negative predication),  
i.e., a derivation from the open assumption of $\varphi^{n}_{0}\alpha_{1} ... \alpha_{n}$ 
(resp. $-\varphi^{n}_{0}\alpha_{1} ... \alpha_{n}$) is an 
$\mathcal{S}_{b}$-derivation.

\textit{Induction step}. 
If $\mathcal{D}_{i}$, for $i \in \{0, ..., n\}$, are 
$\mathcal{S}_{b}$-derivations, then an 
$\mathcal{S}_{b}$-derivation can be constructed by means of the I/E-rules for $as$ 
and $-as$ displayed above.  
\end{definition}

\begin{remark}
The term assumptions are, so to speak, proof-theoretic semantic values of the non-logical constants. 
Applications of the subatomic introduction rules $as$I and $-as$I serve to establish, 
on the basis of these values, the truth 
of atomic sentences and negative predications, respectively. 
\end{remark}

\subsubsection{Bipredicational subatomic identity systems}

\begin{definition}
Atomic sentences $\varphi(\alpha_{1})$ and $\varphi(\alpha_{2})$ are  
\textit{mirror atomic sentences} if and only if they are exactly alike except that the 
former contains occurrences of $\alpha_{1}$ at all the places at which the latter 
contains occurrences of $\alpha_{2}$, and vice versa. 
\end{definition}

\begin{definition}
A \textit{bipredicational subatomic identity system} 
$\mathcal{S}^{=}_{b}$ is a 3-tuple 
$\langle\mathcal{I}, \mathcal{R}_{b}, \mathcal{R}^{=}_{b}\rangle$, 
which extends a bipredicational subatomic system with 
a set $\mathcal{R}^{=}_{b}$ of \textit{I/E-rules for (positive/negative)   
qualified identity sentences}, 
where $\mathcal{Q}\subseteq\mathcal{P}$. 
\begin{enumerate}
\item $\overset{+}{=}_{\mathcal{Q}}$:  
\begin{center}
{\small{
\alwaysNoLine
\AXC{$[\varphi_{1}(\alpha_{1})]^{(1_{1})}$\quad  
$[\varphi_{1}(\alpha_{2})]^{(1_{2})}$}
\def\extraVskip{2pt}
\UIC{$\mathcal{D}_{1_{1}}$ \quad\quad\quad $\mathcal{D}_{1_{2}}$}
\noLine
\UIC{$\varphi_{1}(\alpha_{2})$ \quad\quad $\varphi_{1}(\alpha_{1})$}
\alwaysNoLine
\AXC{$ $}
\def\extraVskip{2pt}
\UIC{$ $}
\UIC{$...$}
\alwaysNoLine
\AXC{$[\varphi_{k}(\alpha_{1})]^{(k_{1})}$ \quad 
$[\varphi_{k}(\alpha_{2})]^{(k_{2})}$}
\def\extraVskip{2pt}
\UIC{$\mathcal{D}_{k_{1}}$ \quad\quad\quad $\mathcal{D}_{k_{2}}$}
\noLine
\UIC{$\varphi_{k}(\alpha_{2})$ \quad\quad $\varphi_{k}(\alpha_{1})$}
\alwaysSingleLine
{\tiny{\RightLabel{($\overset{+}{=}_{\mathcal{Q}}$I), $1_{1}, ..., k_{2}$}}}
\TIC{$\alpha_{1}\overset{+}{=}_{\mathcal{Q}}\alpha_{2}$}
\DP

\medskip

\AXC{$\mathcal{D}_{1}$}
\noLine
\UIC{$\alpha_{1} \overset{+}{=}_{\mathcal{Q}}\alpha_{2}$}
\AXC{$\mathcal{D}_{i_{1}}$}
\noLine
\UIC{$\varphi_{i}(\alpha_{1})$}
\RightLabel{($\overset{+}{=}_{\mathcal{Q}}$E$_{i}$1)}
\BIC{$\varphi_{i}(\alpha_{2})$}
\DP\qquad\AXC{$\mathcal{D}_{1}$}
\noLine
\UIC{$\alpha_{1} \overset{+}{=}_{\mathcal{Q}}\alpha_{2}$}
\AXC{$\mathcal{D}_{i_{2}}$}
\noLine
\UIC{$\varphi_{i}(\alpha_{2})$}
\RightLabel{($\overset{+}{=}_{\mathcal{Q}}$E$_{i}$2)} 
\BIC{$\varphi_{i}(\alpha_{1})$}
\DP
}}
\end{center}
where $\varphi_{i} \in \mathcal{Q}$, 
$i \in \{1, ..., k\}$, 
and $\varphi_{i}(\alpha_{1})$ and $\varphi_{i}(\alpha_{2})$ are 
mirror atomic sentences.

\item $\overset{-}{=}_{\mathcal{Q}}$: 
\begin{center}
{\small{
\alwaysNoLine
\AXC{$[-\varphi_{1}(\alpha_{1})]^{(1_{1})}$\quad  
$[-\varphi_{1}(\alpha_{2})]^{(1_{2})}$}
\def\extraVskip{2pt}
\UIC{$\mathcal{D}_{1_{1}}$ \quad\quad\quad $\mathcal{D}_{1_{2}}$}
\noLine
\UIC{$-\varphi_{1}(\alpha_{2})$ \quad\quad $-\varphi_{1}(\alpha_{1})$}
\alwaysNoLine
\AXC{$ $}
\def\extraVskip{2pt}
\UIC{$ $}
\UIC{$...$}
\alwaysNoLine
\AXC{$[-\varphi_{k}(\alpha_{1})]^{(k_{1})}$ \quad 
$[-\varphi_{k}(\alpha_{2})]^{(k_{2})}$}
\def\extraVskip{2pt}
\UIC{$\mathcal{D}_{k_{1}}$ \quad\quad\quad $\mathcal{D}_{k_{2}}$}
\noLine
\UIC{$-\varphi_{k}(\alpha_{2})$ \quad\quad $-\varphi_{k}(\alpha_{1})$}
\alwaysSingleLine
{\tiny{\RightLabel{($\overset{-}{=}_{\mathcal{Q}}$I), $1_{1}, ..., k_{2}$}}}
\TIC{$\alpha_{1}\overset{-}{=}_{\mathcal{Q}}\alpha_{2}$}
\DP

\medskip

\AXC{$\mathcal{D}_{1}$}
\noLine
\UIC{$\alpha_{1} \overset{-}{=}_{\mathcal{Q}}\alpha_{2}$}
\AXC{$\mathcal{D}_{i_{1}}$}
\noLine
\UIC{$-\varphi_{i}(\alpha_{1})$}
\RightLabel{($\overset{-}{=}_{\mathcal{Q}}$E$_{i}$1)}
\BIC{$-\varphi_{i}(\alpha_{2})$}
\DP 
\qquad\AXC{$\mathcal{D}_{1}$}
\noLine
\UIC{$\alpha_{1} \overset{-}{=}_{\mathcal{Q}}\alpha_{2}$}
\AXC{$\mathcal{D}_{i_{2}}$}
\noLine
\UIC{$-\varphi_{i}(\alpha_{2})$}
\RightLabel{($\overset{-}{=}_{\mathcal{Q}}$E$_{i}$2)} 
\BIC{$-\varphi_{i}(\alpha_{1})$}
\DP
}}
\end{center}
where $\varphi_{i} \in \mathcal{Q}$, 
$i \in \{1, ..., k\}$, 
and $\varphi_{i}(\alpha_{1})$ and $\varphi_{i}(\alpha_{2})$ are 
mirror atomic sentences.
\end{enumerate}
\end{definition}

\begin{remark}
In contrast to the standard I-rules for identity, the I-rules for qualified identity 
allow one to introduce formulae in which the identity predicate is not necessarily 
flanked by two occurrences of the same constant. 
Note that these rules reflect the definitions of the qualified identity predicates. 
\end{remark}

\begin{definition}
It will sometimes be convenient to use the notation $\{\mathcal{D}\}$ 
for the set of the subderivations 
$\mathcal{D}_{2_{1}}, \mathcal{D}_{2_{2}}, ..., \mathcal{D}_{k_{1}}, \mathcal{D}_{k_{2}}$ 
in applications of I-rules for qualified identity. 
\end{definition}

\subsubsection{Bipredicational subatomic natural deduction systems}

\begin{definition}
\textit{Derivations in $\textbf{I0}(\mathcal{S}^{=}_{b})$-systems}. 

\textit{Basic step}. 
Any derivation in an $\mathcal{S}^{=}_{b}$-system and any formula $A$ 
(i.e., a derivation from the open assumption of $A$) is a derivation in an  
$\textbf{I0}(\mathcal{S}^{=}_{b})$-system.

\textit{Induction step}. 
If $\mathcal{D}_{1}$, $\mathcal{D}_{2}$, and $\mathcal{D}_{3}$ are 
derivations in an $\textbf{I0}(\mathcal{S}^{=}_{b})$-system, 
and $C$ possibly a term assumption, then a derivation in an 
$\textbf{I0}(\mathcal{S}^{=}_{b})$-system can be constructed by means of the rules: 
\begin{center}
{\small{ 
\AXC{$\mathcal{D}_{1}$}
\noLine
\UIC{$A$}
\AXC{$\mathcal{D}_{2}$}
\noLine
\UIC{$B$}
\RightLabel{($\&$I)}
\BIC{$A \& B$}
\DP\quad\AXC{$\mathcal{D}_{1}$}
\noLine
\UIC{$A \& B$}
\RightLabel{($\&$E1)}
\UIC{$A$}
\DP\quad\AXC{$\mathcal{D}_{1}$}
\noLine
\UIC{$A \& B$}
\RightLabel{($\&$E2)}
\UIC{$B$}
\DP\quad\AXC{$\mathcal{D}_{1}$}
\noLine
\UIC{$A$}
\RightLabel{($\vee$I1)}
\UIC{$A \vee B$}
\DP\quad\AXC{$\mathcal{D}_{1}$}
\noLine
\UIC{$B$}
\RightLabel{($\vee$I2)}
\UIC{$A \vee B$}
\DP
}}
\end{center}
\begin{center}
{\small{ 
\AXC{$\mathcal{D}_{1}$}
\noLine
\UIC{$A \vee B$}
\AXC{$[A]^{(u)}$}
\noLine
\UIC{$\mathcal{D}_{2}$}
\noLine
\UIC{$C$}
\AXC{$[B]^{(v)}$}
\noLine
\UIC{$\mathcal{D}_{3}$}
\noLine
\UIC{$C$}
\RightLabel{($\vee$E), $u, v$}
\TIC{$C$}
\DP\quad\AXC{\quad $[A]^{(u)}$}
\noLine
\UIC{$\mathcal{D}_{1}$}
\noLine
\UIC{$B$}
\RightLabel{($\supset$I), $u$} 
\UIC{$A \supset B$}
\DP\quad\AXC{$\mathcal{D}_{1}$}
\noLine
\UIC{$A \supset B$}
\AXC{$\mathcal{D}_{2}$}
\noLine
\UIC{$A$}
\RightLabel{($\supset$E)}
\BIC{$B$}
\DP
}}
\end{center}
\begin{center}
{\small{ 
\AXC{$\mathcal{D}_{1}$}
\noLine
\UIC{$A(x/o)$}
\RightLabel{($\forall$I)}
\UIC{$\forall x A$}
\DP\quad\AXC{$\mathcal{D}_{1}$}
\noLine
\UIC{$\forall x A$}
\RightLabel{($\forall$E)}
\UIC{$A(x/o)$}
\DP\quad\AXC{$\mathcal{D}_{1}$}
\noLine
\UIC{$A(x/o)$}
\RightLabel{($\exists$I)}
\UIC{$\exists x A$}
\DP\quad\AXC{$\mathcal{D}_{1}$}
\noLine
\UIC{$\exists x A$} \AXC{$ $}
\noLine
\UIC{$[A(x/o)]^{(u)}$}
\noLine
\UIC{$\mathcal{D}_{2}$}
\noLine
\UIC{$C$}
\RightLabel{($\exists$E), $u$}
\BIC{$C$}
\DP
}}
\end{center}
\begin{center}
{\small{ 
\quad\AXC{$\mathcal{D}_{1}$} 
\noLine
\UIC{$\bot$}
\RightLabel{($\bot$i)} 
\UIC{$A$}
\DP 
}}
\end{center}
Side conditions: 
\begin{enumerate}
\item In $\forall$I: 
(i) if $o$ is a proper variable $y$, then $o \equiv x$ or $o$ is not free in $A$, 
and $o$ is not free in any assumption of a formula which is open in the derivation 
of $A(x/o)$; 
(ii) if $o$ is a nominal constant, then $o$ does neither occur in an 
undischarged assumption of a formula, nor in $\forall x A$, nor in a 
term assumption leaf $o\Gamma$;
(iii) $o$ is  nominal constant and \AXC{$\mathcal{D}_{1}$}
\noLine
\UIC{$A(x/o)$}
\DP for all $o \in \mathcal{C}$. 

\item In $\forall$E: $o$ is free for $x$ in $A$. 

\item In $\exists$E: 
(i) if $o$ is a proper variable $y$, then $o \equiv x$ or $o$ is not free in $A$, 
and $o$ is not free in $C$ nor in any assumption of a formula which is open in 
the derivation of the upper occurrence of $C$ other than $[A(x/o)]^{(u)}$; 
(ii) if $o$ is a nominal constant, then $o$ does neither occur in an undischarged assumption 
of a formula, nor in $\exists x A$, nor in $C$, nor in a term assumption leaf $o\Gamma$. 

\item In $\exists$I: $o$ is free for $x$ in $A$.
\end{enumerate}
Minimal bipredicational subatomic natural deduction systems, 
$\textbf{M0}(\mathcal{S}^{=}_{b})$-systems, 
result from 
$\textbf{I0}(\mathcal{S}^{=}_{b})$-systems,  
in case $\bot$i is removed. 

\smallskip
In case we employ the $\forall$I-rule according to the provisos 
for it given in (i) [(ii), (iii)], we use the labels $\forall$I.i [$\forall$I.ii, $\forall$I.iii]. 
Similarly, for the $\exists$E-rule and the labels $\exists$E.i and $\exists$E.ii. 
\end{definition}

\subsection{Bipredicational natural deduction for qualified definiteness}

We now add rules for the 
introduction and elimination of qualified definiteness 
to $\textbf{I0}(\mathcal{S}^{=}_{b})$-systems 
in order to obtain $\textbf{I0}(\mathcal{S}^{=}_{b})\iota$-systems 
which are sufficient to define a proof-theoretic semantics for the 
simplest possible constructions involving definite descriptions. 

\begin{definition}
Let $\mathcal{Q} \subseteq \mathcal{P}$. 
In the $\iota_{\mathcal{Q}}$I-rule below, the conclusion of $\mathcal{D}_{1}$ 
[$\mathcal{D}_{2}$, $\mathcal{D}_{3}$] corresponds to 
the (E)- [(QU)-, (P)-] clause. 
Likewise for $\iota_{\mathcal{Q}}-$I.
\begin{enumerate}
\item \textit{Rules for positive qualified definiteness}: 
\begin{center}
{\small{
\AXC{$\mathcal{D}_{1}$}
\noLine
\UIC{$\exists x \varphi(x)$}
\AXC{$\mathcal{D}_{2}$}
\noLine
\UIC{$\forall u \forall v((\varphi(u)$ $\&$ $\varphi(v)) \supset u \overset{+}{=}_{\mathcal{Q}} v)$}
\AXC{$\mathcal{D}_{3}$}
\noLine
\UIC{$\forall w(\varphi(w) \supset \psi(w))$}
\RightLabel{($\iota_{\mathcal{Q}}$I)}
\TIC{$\psi(\iota_{\mathcal{Q}} x\varphi(x))$}
\DP
}}
\end{center}

\begin{center}
{\small{
\AXC{$\mathcal{D}_{1}$}
\noLine
\UIC{$\psi(\iota_{\mathcal{Q}} x\varphi(x))$}
\RightLabel{($\iota_{\mathcal{Q}}$E1)}
\UIC{$\exists x \varphi(x)$}
\DP\quad\AXC{$\mathcal{D}_{1}$}
\noLine
\UIC{$\psi(\iota_{\mathcal{Q}} x\varphi(x))$}
\RightLabel{($\iota_{\mathcal{Q}}$E2)}
\UIC{$\forall u \forall v((\varphi(u)$ $\&$ $\varphi(v)) \supset u \overset{+}{=}_{\mathcal{Q}} v)$}
\DP\quad\AXC{$\mathcal{D}_{1}$}
\noLine
\UIC{$\psi(\iota_{\mathcal{Q}} x\varphi(x))$}
\RightLabel{($\iota_{\mathcal{Q}}$E3)}
\UIC{$\forall w(\varphi(w) \supset \psi(w))$}
\DP
}}
\end{center}
The $\iota_{\mathcal{Q}}$I/E-rules for $-\psi(\iota_{\mathcal{Q}} x\varphi(x))$ are analogous. 

\item \textit{Rules for negative qualified definiteness}: 
\begin{center}
{\small{
\AXC{$\mathcal{D}_{1}$}
\noLine
\UIC{$\exists x -\varphi(x)$}
\AXC{$\mathcal{D}_{2}$}
\noLine
\UIC{$\forall u \forall v((-\varphi(u)$ $\&$ $-\varphi(v)) \supset u \overset{-}{=}_{\mathcal{Q}} v)$}
\AXC{$\mathcal{D}_{3}$}
\noLine
\UIC{$\forall w(-\varphi(w) \supset \psi(w))$}
\RightLabel{($\iota_{\mathcal{Q}}-$I)}
\TIC{$\psi(\iota_{\mathcal{Q}} x-\varphi(x))$}
\DP
}}
\end{center}

\begin{center}
{\small{
\AXC{$\mathcal{D}_{1}$}
\noLine
\UIC{$\psi(\iota_{\mathcal{Q}} x-\varphi(x))$}
\RightLabel{($\iota_{\mathcal{Q}}-$E1)}
\UIC{$\exists x -\varphi(x)$}
\DP\quad\AXC{$\mathcal{D}_{1}$}
\noLine
\UIC{$\psi(\iota_{\mathcal{Q}} x-\varphi(x))$}
\RightLabel{($\iota_{\mathcal{Q}}-$E2)}
\UIC{$\forall u \forall v((-\varphi(u)$ $\&$ $-\varphi(v)) \supset u \overset{-}{=}_{\mathcal{Q}} v)$}
\DP\quad\AXC{$\mathcal{D}_{1}$}
\noLine
\UIC{$\psi(\iota_{\mathcal{Q}} x-\varphi(x))$}
\RightLabel{($\iota_{\mathcal{Q}}-$E3)}
\UIC{$\forall w(-\varphi(w) \supset \psi(w))$}
\DP
}}
\end{center}
The $\iota_{\mathcal{Q}}-$I/E-rules for $-\psi(\iota_{\mathcal{Q}} x-\varphi(x))$ are analogous. 
\end{enumerate}
\end{definition}

\begin{example}
Let $\mathcal{Q} = \{\varphi_{1}, ..., \varphi_{k}\}$, 
$\mathcal{Q} \subseteq \mathcal{P}$, and 
$\varphi_{i}, \varphi_{j} \in \mathcal{Q}$, 
where $i, j \in \{1, ..., k\}$ and $i \not= j$.
\begin{equation}
{\small{
\AXC{$\varphi_{i}\Gamma$}
\AXC{$...$}
\AXC{$\alpha\Gamma$}
\TIC{$\varphi_{i}(\alpha)$}
\LeftLabel{$\mathcal{D}_{1} =$}
\UIC{$\exists x \varphi_{i}(x)$}
\DP
}}
\end{equation}

\begin{equation}
{\small{
\AXC{$[\varphi_{1}(\alpha)]^{(1_{1})}$}
\UIC{$\varphi_{1}\Gamma$}
\AXC{$...$}
\AXC{$[\varphi_{i}(\alpha) \& \varphi_{i}(\beta)]^{(1)}$}
\UIC{$\varphi_{i}(\beta)$}
\UIC{$\beta\Gamma$}
\TIC{$\varphi_{1}(\beta)$}
\AXC{$[\varphi_{1}(\beta)]^{(1_{2})}$}
\UIC{$\varphi_{1}\Gamma$}
\AXC{$...$}
\AXC{$[\varphi_{i}(\alpha) \& \varphi_{i}(\beta)]^{(1)}$}
\UIC{$\varphi_{i}(\alpha)$}
\UIC{$\alpha\Gamma$}
\TIC{$\varphi_{1}(\alpha)$}
\AXC{$\{\mathcal{D}\}$}
\RightLabel{$1_{1}, ..., k_{2}$}
\TIC{$\alpha \overset{+}{=}_{\mathcal{Q}} \beta$}
\RightLabel{$1$}
\UIC{$(\varphi_{i}(\alpha) \& \varphi_{i}(\beta)) \supset \alpha \overset{+}{=}_{\mathcal{Q}} \beta$}
\RightLabel{iii}
\UIC{$\forall v((\varphi_{i}(\alpha) \& \varphi_{i}(v)) \supset \alpha \overset{+}{=}_{\mathcal{Q}} v)$}
\LeftLabel{$\mathcal{D}_{2} =$}
\RightLabel{iii}
\UIC{$\forall u\forall v((\varphi_{i}(u) \& \varphi_{i}(v)) \supset u \overset{+}{=}_{\mathcal{Q}} v)$}
\DP
}}
\end{equation}

\begin{equation}
{\small{
\AXC{$\varphi_{j}\Gamma$}
\AXC{$...$}
\AXC{$[\varphi_{i}(\alpha)]^{(2)}$}
\UIC{$\alpha\Gamma$}
\TIC{$\varphi_{j}(\alpha)$}
\RightLabel{2}
\UIC{$\varphi_{i}(\alpha) \supset \varphi_{j}(\alpha)$}
\LeftLabel{$\mathcal{D}_{3} =$}
\RightLabel{iii}
\UIC{$\forall w(\varphi_{i}(w) \supset \varphi_{j}(w))$}
\DP
}}
\end{equation}

\begin{equation}
{\small{
\AXC{$\mathcal{D}_{1}$}
\noLine
\UIC{$\exists x \varphi_{i}(x)$}
\AXC{$\mathcal{D}_{2}$}
\noLine
\UIC{$\forall u\forall v((\varphi_{i}(u) \& \varphi_{i}(v)) \supset u \overset{+}{=}_{\mathcal{Q}} v)$}
\AXC{$\mathcal{D}_{3}$}
\noLine
\UIC{$\forall w(\varphi_{i}(w) \supset \varphi_{j}(w))$}
\RightLabel{($\iota_{\mathcal{Q}}$I)}
\TIC{$\varphi_{j}(\iota_{\mathcal{Q}} x\varphi_{i}(x))$}
\DP
}}
\end{equation}
\end{example}

\subsection{Normalization and the subformula property}

Normalization and the subformula property 
for $\textbf{I0}(\mathcal{S}^{=}_{b})$-systems have been established in \cite{negpredd} 
making use of the methods developed in \cite{Prawitz}; see also \cite{TS}. 
These results guarantee, e.g., the consistency of the systems and simplify proof search in them. 

In order to prove normalization for $\textbf{I0}(\mathcal{S}^{=}_{b})\iota$-systems,  
we make use of the following conversions. 

\begin{definition}
The \textit{conversions (detour, permutation, simplification)}  
for $\textbf{I0}(\mathcal{S}^{=}_{b})\iota$-systems 
comprise those for $\textbf{I0}(\mathcal{S}^{=}_{b})$-systems (see \cite{negpredd}) 
and the following detour conversions:   
\begin{enumerate}
\item \textit{$\iota_{\mathcal{Q}}$-Conversions}:
\begin{center}
{\small{
\AXC{$\mathcal{D}_{1}$}
\noLine
\UIC{$\exists x \varphi(x)$}
\AXC{$\mathcal{D}_{2}$}
\noLine
\UIC{$\forall u \forall v((\varphi(u)$ $\&$ $\varphi(v)) \supset u \overset{+}{=}_{\mathcal{Q}} v)$}
\AXC{$\mathcal{D}_{3}$}
\noLine
\UIC{$\forall w(\varphi(w) \supset \psi(w))$}
\RightLabel{($\iota_{\mathcal{Q}}$I)}
\TIC{$\psi(\iota_{\mathcal{Q}} x\varphi(x))$}
\RightLabel{($\iota_{\mathcal{Q}}$E1)}
\UIC{$\exists x \varphi(x)$}
\DP\quad conv\quad\AXC{$\mathcal{D}_{1}$}
\noLine
\UIC{$\exists x \varphi(x)$}
\DP
}}
\end{center}
\begin{center}
{\small{
\AXC{$\mathcal{D}_{1}$}
\noLine
\UIC{$\exists x \varphi(x)$}
\AXC{$\mathcal{D}_{2}$}
\noLine
\UIC{$\forall u \forall v((\varphi(u)$ $\&$ $\varphi(v)) \supset u \overset{+}{=}_{\mathcal{Q}} v)$}
\AXC{$\mathcal{D}_{3}$}
\noLine
\UIC{$\forall w(\varphi(w) \supset \psi(w))$}
\RightLabel{($\iota_{\mathcal{Q}}$I)}
\TIC{$\psi(\iota_{\mathcal{Q}} x\varphi(x))$}
\RightLabel{($\iota_{\mathcal{Q}}$E2)}
\UIC{$\forall u \forall v((\varphi(u)$ $\&$ $\varphi(v)) \supset u \overset{+}{=}_{\mathcal{Q}} v)$}
\DP

\medskip
conv

\medskip
\AXC{$\mathcal{D}_{2}$}
\noLine
\UIC{$\forall u \forall v((\varphi(u)$ $\&$ $\varphi(v)) \supset u \overset{+}{=}_{\mathcal{Q}} v)$}
\DP
}}
\end{center}
\begin{center}
{\small{
\AXC{$\mathcal{D}_{1}$}
\noLine
\UIC{$\exists x \varphi(x)$}
\AXC{$\mathcal{D}_{2}$}
\noLine
\UIC{$\forall u \forall v((\varphi(u)$ $\&$ $\varphi(v)) \supset u \overset{+}{=}_{\mathcal{Q}} v)$}
\AXC{$\mathcal{D}_{3}$}
\noLine
\UIC{$\forall w(\varphi(w) \supset \psi(w))$}
\RightLabel{($\iota_{\mathcal{Q}}$I)}
\TIC{$\psi(\iota_{\mathcal{Q}} x\varphi(x))$}
\RightLabel{($\iota_{\mathcal{Q}}$E3)}
\UIC{$\forall w(\varphi(w) \supset \psi(w))$}
\DP

\medskip
conv

\medskip
\AXC{$\mathcal{D}_{3}$}
\noLine
\UIC{$\forall w(\varphi(w) \supset \psi(w))$}
\DP
}}
\end{center}

\item \textit{$\iota_{\mathcal{Q}}-$-Conversions}: 
	analogous. 
\end{enumerate}
\end{definition}

\begin{remark}
Unlike the $\iota$E2-rules in \cite{iota1}, \cite{iota2}, the above E2-rules have a single premiss 
and invert directly. 
\end{remark}

\begin{theorem}
\textit{Normalization}:  
Any derivation $\mathcal{D}$ in an $\textbf{I0}(\mathcal{S}^{=}_{b})\iota$-system can be transformed 
into a normal $\textbf{I0}(\mathcal{S}^{=}_{b})\iota$-derivation. 
\end{theorem}
\begin{proof}
We repeat the corresponding proof for 
$\textbf{I0}(\mathcal{S}^{=}_{b})$-systems in \cite{negpredd},  
taking also the detour conversions for qualified definiteness into account.  
As a result, all detours can be eliminated from derivations in these systems. 
\end{proof}

Importantly, $\textbf{I0}(\mathcal{S}^{=}_{b})\iota$-systems enjoy the subformula property 
as a special case of the subexpression property. 
The latter property deals with units and expressions. 
Roughly, a unit is either a formula or a term assumption $\tau\Gamma$, and 
an expression is either a formula or the non-logical constant $\tau$ of $\tau\Gamma$. 
\begin{theorem}
\textit{Subexpression property}: 
If $\mathcal{D}$ is a normal derivation of a unit $U$ from a set of 
units $\Gamma$ in an $\textbf{I0}(\mathcal{S}^{=}_{b})\iota$-system, 
then each unit in $\mathcal{D}$ is a subexpression of an expression in $\Gamma\cup\{U\}$. 
\end{theorem}
\begin{proof}
We proceed like in the corresponding proof for $\textbf{I0}(\mathcal{S}^{=}_{b})$-systems 
in \cite{negpredd}.  
As a result, all expressions in $\mathcal{D}$ are subexpressions of either the 
root or the leaves of $\mathcal{D}$. 
\end{proof}

\begin{corollary}
\textit{Subformula property}: 
If $\mathcal{D}$ is a normal $\textbf{I0}(\mathcal{S}^{=}_{b})\iota$-derivation 
of formula $A$ from a set of formulae $\Gamma$, then 
each formula in $\mathcal{D}$ is a subformula of a formula in $\Gamma\cup\{A\}$. 
\end{corollary}

\begin{remark}
Since the identity predicates used in the proof systems \cite{iota1}, \cite{iota2}, are primitive, 
such a subformula result is not available for these systems. 
This remark also applies to other available intuitionistic natural deduction systems for definiteness 
(e.g., \cite{NKiota1}, \cite{Stenlund}). 
\end{remark}

\begin{corollary}
\textit{Internal completeness}. 
Internal completeness in the sense of \cite{GirardLudics} (pp. 139--140)
is given by Corollary 3.1. 
To establish internal completeness for $\textbf{I0}(\mathcal{S}^{=}_{b})\iota$-systems 
in the sense of \cite{negpredd} (p. 127), we proceed like described therein.  
\end{corollary}

\section{A proof-theoretic semantics}

On the basis of the results obtained, we may 
formulate a subatomic proof-theoretic semantics for qualified definiteness. 
For this purpose, we adjust the corresponding definitions form \cite{negpredd} 
to the present systems. 

\begin{definition}
\begin{enumerate}
\item A derivation $\mathcal{D}$ of a formula $A$ in an $\textbf{I0}(\mathcal{S}^{=}_{b})\iota$-system is 
a \textit{canonical derivation} iff it derives $A$ by means of an application of an I-rule 
(in the last step of $\mathcal{D}$). 

\item A canonical derivation $\mathcal{D}$ of $A$ in an $\textbf{I0}(\mathcal{S}^{=}_{b})\iota$-system is 
a \textit{canonical proof} 
of $A$ in that system 
iff there are no applications of $as$-rules or $-as$-rules in $\mathcal{D}$ and 
all assumptions of $\mathcal{D}$ have been discharged. 

\item The conclusions of canonical $\textbf{I0}(\mathcal{S}^{=}_{b})\iota$-derivations 
are $\textbf{I0}(\mathcal{S}^{=}_{b})\iota$-\textit{theses} and the conclusions of 
$\textbf{I0}(\mathcal{S}^{=}_{b})\iota$-derivations which are also proofs 
are $\textbf{I0}(\mathcal{S}^{=}_{b})\iota$-\textit{theorems}. 
\end{enumerate}
\end{definition}

\begin{definition} 
\textit{Meaning}:  
Let $I$ be an $\textbf{I0}(\mathcal{S}^{=}_{b})\iota$-system.   
\begin{enumerate}
\item  
The meaning of a \textit{non-logical constant} $\tau$ is given by the 
term assumptions $\tau\Gamma$ for $\tau$ which are 
determined by the subatomic base of the $\mathcal{S}^{=}_{b}$-system of $I$.  
\item  
The meaning of a \textit{formula} $A$ of $\mathcal{L}\iota$ is given by the set of canonical derivations 
of $A$ in $I$. 
\end{enumerate}
\end{definition}

\begin{remark}
The rules for qualified identity defined in \cite{negpredd} allow 
not only for reductions in terms of conversions, but also for expansions 
(cf. \cite{naturalistic}: 256). 
This is a further point, in which 
they differ from the standard natural deduction rules for identity (cf. \cite{negpredd}: 104). 
For an overview of the structural proof theory of identity see \cite{AIpsh}. 
\end{remark}

\begin{remark} 
Note that this formal account of meaning does not make use of 
a semantic ontology (e.g., individuals, possible worlds),  
something essential to model-theoretic semantics. 
Specifically, the meaning of $\exists$-formulae does not presuppose a domain of individuals. 
Strictly speaking, $\exists xA$ reads: 
\textquoteleft For at least one $x$, $A$\textquoteright, where $x$ is a nominal variable 
ranging over $\mathcal{C}$. 
This feature of the present semantics makes it particularly natural for the analysis 
of constructions which involve non-denoting (or empty) terms 
(e.g., \textquoteleft Pegasus\textquoteright, \textquoteleft the captive unicorn\textquoteright).  
\end{remark}

\section{On incomplete descriptions}

Qualified uniqueness allows for fine-tuning. 
\begin{remark}
Let $\{\varphi_{i}\} \subset \mathcal{Q}^{\prime} \subset \mathcal{P}$ 
and $\varphi_{i} \in \mathcal{P}$, where $i\in \{1, ..., k\}$.   
We consider the following cases: 
(i) $\mathcal{Q} = \mathcal{P}$, 
(ii) $\mathcal{Q} = \mathcal{Q}^{\prime}$, and 
(iii) $\mathcal{Q} = \{\varphi_{i}\}$.  

Case (i): Like (2),  
but with $\mathcal{Q}$ replaced by $\mathcal{P}$. 
This case gives us the maximal degree of qualified uniqueness. 
For every $x$ and $y$, if they are $\varphi_{i}$, then they are identical 
with respect to every predicate  
(i.e., they are indiscernible in every respect). 

Case (ii):  
Like case (i), but with 
$\mathcal{P}$ replaced by $\mathcal{Q}^{\prime}$ and with 
$\{\mathcal{D}\}$ replaced by $\{\mathcal{D}\}^{\prime}$, 
where $\{\mathcal{D}\}^{\prime} \subset \{\mathcal{D}\}$. 
This case gives us an intermediate degree of qualified uniqueness.  
For every $x$ and $y$, if they are $\varphi_{i}$, then they are identical 
with respect to every predicate in $\mathcal{Q}^{\prime}$ 
(i.e., they are indiscernible with respect to $\mathcal{Q}^{\prime}$, 
but discernible with respect to $\mathcal{P}\setminus\mathcal{Q}^{\prime}$).

Case (iii): 
\begin{equation}
{\small{
\AXC{$[\varphi_{i}(\alpha)]^{(1_{1})}$}
\UIC{$\varphi_{i}\Gamma$}
\AXC{$...$}
\AXC{$[\varphi_{i}(\alpha) \& \varphi_{i}(\beta)]^{(1)}$}
\UIC{$\varphi_{i}(\beta)$}
\UIC{$\beta\Gamma$}
\TIC{$\varphi_{i}(\beta)$}
\AXC{$[\varphi_{i}(\beta)]^{(1_{2})}$}
\UIC{$\varphi_{i}\Gamma$}
\AXC{$...$}
\AXC{$[\varphi_{i}(\alpha) \& \varphi_{i}(\beta)]^{(1)}$}
\UIC{$\varphi_{i}(\alpha)$}
\UIC{$\alpha\Gamma$}
\TIC{$\varphi_{i}(\alpha)$}
\RightLabel{$1_{1}, 1_{2}$}
\BIC{$\alpha \overset{+}{=}_{\{\varphi_{i}\}} \beta$}
\RightLabel{$1$}
\UIC{$(\varphi_{i}(\alpha) \& \varphi_{i}(\beta)) \supset \alpha \overset{+}{=}_{\{\varphi_{i}\}} \beta$}
\RightLabel{iii}
\UIC{$\forall y((\varphi_{i}(\alpha) \& \varphi_{i}(y)) \supset \alpha \overset{+}{=}_{\{\varphi_{i}\}} y)$}
\RightLabel{iii}
\UIC{$\forall x\forall y((\varphi_{i}(x) \& \varphi_{i}(y)) \supset x \overset{+}{=}_{\{\varphi_{i}\}} y)$}
\DP
}}
\end{equation}
This case gives us the minimal degree of qualified uniqueness.  
For every $x$ and $y$, if they are $\varphi_{i}$, then they are identical 
with respect to every predicate in the singleton $\{\varphi_{i}\}$  
(i.e., they are indiscernible with respect to the predicate $\varphi_{i}$, 
but discernible with respect to any other predicate in $\mathcal{P}\setminus\{\varphi_{i}\}$). 
(Likewise for negative qualified uniqueness.) 
\end{remark}

Qualified definiteness allows for fine-tuning, since it involves qualified uniqueness. 
\begin{remark}
Let $\{\varphi_{i}\} \subset \mathcal{Q}^{\prime} \subset \mathcal{P}$,  
let $P = \varphi_{i}$, and $B = \varphi_{j}$ 
for $\varphi_{i}, \varphi_{j} \in \mathcal{Q}^{\prime}$, 
where $i, j \in \{1, ..., k\}$ and $i \not= j$. 
$P$: \textquoteleft ... is a pope\textquoteright;  
$B$: \textquoteleft ... is bald\textquoteright. 
And let 
$\mathcal{D}_{2}(i)$ [$\mathcal{D}_{2}(ii)$, $\mathcal{D}_{2}(iii)$] 
refer to the derivation for case (i) [(ii), (iii)] mentioned in the previous remark. 
We may, then, distinguish three general cases of qualified definiteness. 

Case (i). Maximal qualified definiteness:  
\begin{equation}
{\small{
\AXC{$\mathcal{D}_{1}$}
\noLine
\UIC{$\exists x \varphi_{i}(x)$}
\AXC{$\mathcal{D}_{2(i)}$}
\noLine
\UIC{$\forall u\forall v((\varphi_{i}(u) \& \varphi_{i}(v)) \supset u \overset{+}{=}_{\mathcal{P}} v)$}
\AXC{$\mathcal{D}_{3}$}
\noLine
\UIC{$\forall w(\varphi_{i}(w) \supset \varphi_{j}(w))$}
\RightLabel{($\iota_{\mathcal{P}}$I)}
\TIC{$\varphi_{j}(\iota_{\mathcal{P}} x\varphi_{i}(x))$}
\DP
}}
\end{equation}
The premisses of the $\iota_{\mathcal{P}}$I-application say 
that there is at least one thing which is $\varphi_{i}$, 
that any two things which are $\varphi_{i}$ are the same in any respect, and 
that everything that is $\varphi_{i}$ is $\varphi_{j}$. 
The conclusion 
$\varphi_{j}(\iota_{\mathcal{P}} x\varphi_{i}(x))$ can be read:
\textquoteleft the $\mathcal{P}$-unique $x$ 
	which is $\varphi_{i}$ is $\varphi_{j}$\textquoteright, or,  
simplifying the reading of Definition 2.4(1) further,    
\textquoteleft the only $\varphi_{i}$ is $\varphi_{j}$\textquoteright. 
We may use these readings only in case there is a single $x$ that is $\varphi_{i}$. 
This is definiteness proper. 
We use it for the analysis of (1.1), in case there is no schism. 

Case (ii). Intermediate qualified definiteness:  
\begin{equation}
{\small{
\AXC{$\mathcal{D}_{1}$}
\noLine
\UIC{$\exists x \varphi_{i}(x)$}
\AXC{$\mathcal{D}_{2(ii)}$}
\noLine
\UIC{$\forall u\forall v((\varphi_{i}(u) \& \varphi_{i}(v)) \supset 
u \overset{+}{=}_{\mathcal{Q}^{\prime}} v)$}
\AXC{$\mathcal{D}_{3}$}
\noLine
\UIC{$\forall w(\varphi_{i}(w) \supset \varphi_{j}(w))$}
\RightLabel{($\iota_{\mathcal{Q}^{\prime}}$I)}
\TIC{$\varphi_{j}(\iota_{\mathcal{Q}^{\prime}} x\varphi_{i}(x))$}
\DP
}}
\end{equation}
The premisses of the $\iota_{\mathcal{Q}^{\prime}}$I-application say 
that there is at least one thing which is $\varphi_{i}$, 
that any two things which are $\varphi_{i}$ are the same (only) 
in any $\mathcal{Q}^{\prime}$-respect, and that everything that is $\varphi_{i}$ is $\varphi_{j}$. 
The conclusion 
$\varphi_{j}(\iota_{\mathcal{Q}^{\prime}} x\varphi_{i}(x))$ can be read:
\textquoteleft the $\mathcal{Q}^{\prime}$-unique $x$ 
	which is $\varphi_{i}$ is $\varphi_{j}$\textquoteright, or simply 
\textquoteleft the $\varphi_{i}$ is $\varphi_{j}$\textquoteright. 
We may use these readings only in case there are at least two things 
that are $\varphi_{i}$ 
which are discernible with respect to $\mathcal{P}\setminus\mathcal{Q}^{\prime}$.
It will be natural to use this restricted kind of definiteness for the analysis of (1.1) 
in times of schism. 

Case (iii). Minimal qualified definiteness:  
\begin{equation}
{\small{
\AXC{$\mathcal{D}_{1}$}
\noLine
\UIC{$\exists x \varphi_{i}(x)$}
\AXC{$\mathcal{D}_{2(iii)}$}
\noLine
\UIC{$\forall u\forall v((\varphi_{i}(u) \& \varphi_{i}(v)) \supset u \overset{+}{=}_{\{\varphi_{i}\}} v)$}
\AXC{$\mathcal{D}_{3}$}
\noLine
\UIC{$\forall w(\varphi_{i}(w) \supset \varphi_{j}(w))$}
\RightLabel{($\iota_{\{\varphi_{i}\}}$I)}
\TIC{$\varphi_{j}(\iota_{\{\varphi_{i}\}} x\varphi_{i}(x))$}
\DP
}}
\end{equation}
The premisses of the $\iota_{\{\varphi_{i}\}}$I-application say 
that there is at least one thing which is $\varphi_{i}$, 
that any two things which are $\varphi_{i}$ are the same only with respect to $\{\varphi_{i}\}$, and 
that everything that is $\varphi_{i}$ is $\varphi_{j}$. 
	The conclusion 
$\varphi_{j}(\iota_{\{\varphi_{i}\}} x\varphi_{i}(x))$ can be read:
	\textquoteleft the $\{\varphi_{i}\}$-unique $x$ 
	which is $\varphi_{i}$ is $\varphi_{j}$\textquoteright. 
We may use this reading only in case there are at least two things 
that are $\varphi_{i}$ 
which are discernible with respect to $\mathcal{P}\setminus\{\varphi_{i}\}$.
In a sense, this minimal degree of definiteness comes close to generic definiteness: 
\textquoteleft the generic $\varphi_{i}$ is $\varphi_{j}$\textquoteright\space 
(e.g., \textquoteleft The Englishman is brave\textquoteright).  
Similarly for negative qualified definiteness.  
\end{remark}

\begin{remark}
A negative predication with a definite description: 
\begin{itemize}
\item[(1.3)] The king of France is not real.\\
	$-Real(\iota_{\mathcal{P}} x(King\text{-}of(x, France)))$
\end{itemize}	
Cf. Remark 4.2. 
\end{remark}

\section{Outlook}

Adapting the resources of \cite{iota1}, \cite{iota2} to 
the present framework, we may use it also for the analysis of 
constructions such as, e.g., (1.2), (1.4)-(1.6), 
and further challenging cases discussed in the literature. 
\begin{itemize}
\item[(1.4)] The dog descends from the wolf. (Cf. \cite{Ludlow}: (33).)\\ 
	$Descends$-$from(\iota_{\{Dog\}} x(Dog(x)), \iota_{\{Wolf\}} y(Wolf(y)))$ 

\item[(1.5)] The pope put the zucchetto on the zucchetto. (Cf. \cite{Ludlow}: (38).)\\
	$Put$-$on(
	\iota_{\mathcal{P}} x(Pope(x)), 
	\iota_{\mathcal{Q}^{\prime}} y(Zucchetto(y)), 
	\iota_{\mathcal{Q}^{\prime\prime}} z(Zucchetto(z))
	)$

\item[(1.6)] The man wearing the beret with the button is French. (\cite{Kuhn}: 450.)\\
	$French(\iota_{\mathcal{Q}} x(Man(x)$ $\&$ $Wears(x, \iota_{\mathcal{Q}} y(Beret(y)$ 
	$\&$ $Has(y, \iota_{\mathcal{Q}} z(Button(z)))))))$
\end{itemize}

\nocite{*}
\bibliographystyle{eptcs}
\bibliography{references}

\end{document}